\documentclass{article}
\usepackage[utf8]{inputenc}
\usepackage{amstext,amsmath,amssymb,amsgen,amsfonts,latexsym,theorem,pifont}
\usepackage{mylayout}
\usepackage[dvipdfmx]{graphicx,xcolor}
\usepackage{comment}
\usepackage{url}
\usepackage{oursymbols}
\usepackage{ourtheorems}

\begin{document}
\sloppy

%

\title{
Quantum Query Lower Bounds for Key Recovery Attacks\\ on the Even-Mansour Cipher
}
\author{
  Akinori Kawachi\footnote{Graduate School of Engineering, Mie University}\\
  \texttt{kawachi@info.mie-u.ac.jp}
  \and
  Yuki Naito\footnotemark[1]\\
  \texttt{naitoyuki0713@gmail.com}
}
\date{}
\maketitle
\thispagestyle{plain}
\pagestyle{plain}

\begin{abstract}
The Even-Mansour (EM) cipher is one of the famous constructions for a block cipher. Kuwakado and Morii demonstrated that a quantum adversary can recover its $n$-bit secret keys only with  $\OO(n)$ nonadaptive quantum queries. While the security of the EM cipher and its variants is well-understood for classical adversaries, very little is currently known of their quantum security. Towards a better understanding of the quantum security, or the limits of quantum adversaries for the EM cipher, we study the quantum query complexity for the key recovery of the EM cipher and prove every quantum algorithm requires $\Omega(n)$ quantum queries for the key recovery even if it is allowed to make adaptive queries. Therefore, the quantum attack of Kuwakado and Morii has the optimal query complexity up to a constant factor, and we cannot asymptotically improve it even with adaptive quantum queries.
\end{abstract}

\section{Introduction}\label{sec:intro}
Since the discovery of quantum algorithms for factorization and discrete logarithm problems by Shor \cite{Sho97}, it has become widely known that many practical schemes based on public-key cryptography can be broken by quantum computers theoretically. Although the quantum computer that can be implemented with the current technology does not pose a threat to practical cryptographic schemes, it is essential to study the schemes that are secure enough against quantum computers that will be developed in the near future. 

Much of the early work on quantum attacks focused on public-key cryptosystems, and only generic algorithms based on Grover's quantum search \cite{Gro96} were known to attack symmetric-key cryptosystems. 
However, recent studies have shown that more sophisticated quantum attacks are possible even against some symmetric-key cryptosystems. 
Kuwakado and Morii provided efficient quantum attacks against the well-known symmetric-key primitives 
such as the 3-round Feistel structure \cite{KM10} and the Even-Mansour (EM) cipher \cite{KM12} using Simon's quantum algorithm \cite{Sim97}.
Following their celebrated results, several papers revealed 
new quantum attacks against many symmetric-key constructions
such as the work of Kaplan, Leurent, Leverrier, and Naya-Plasencia \cite{KLL+16} that provided efficient quantum attacks on 
some of the most common block-cipher modes of operations for message authentication and authenticated encryption.
The discovery of these quantum attacks against symmetric-key cryptosystems has led us to focus not only on analyses 
of the potential capabilities of quantum adversaries for public-key cryptography but also on those for symmetric-key cryptography.

In particular, the security of the EM cipher and its variants has been studied in many papers so far against classical and quantum adversaries.
The EM cipher is a well-known construction for block ciphers and has a very simple structure to achieve the security of pseudorandom functions.
For a random public permutation $\pi:\doubleZ_2^n\rightarrow\doubleZ_2^n$ and secret keys $k_1,k_2\in\doubleZ_2^n$, 
its encryption function is defined as $EM(x):=\pi(x + k_1) + k_2$.

The classical security of the EM cipher and its variants has been broadly studied.
The original paper of Even and Mansour proved that classical adversaries 
require $\OO(2^{n/2})$ queries to break the EM cipher \cite{EM97}.
Chen and Steinberger provided query lower bounds for generalizations 
of the EM cipher, called the iterated EM ciphers 
$iEM_t(x):=k_t+\pi_t(k_{t-1}+\pi_{t-1}(\cdots k_1+\pi_1(k_0+x)\cdots))$ \cite{CS14}.
They proved the tight query lower bound of $\Omega(2^{(t/(t+1))n})$ for attacking the variant that matches to query upper bounds of $\OO(2^{(t/(t+1))n})$ by a generalization of Daemen's attack \cite{Dae91}, which was pointed out by Bogdanov, Knudsen, Leander, Standaert, Steinberger, and Tischhauser \cite{BKL+12}.
Chen, Lambooij, and Mennink also studied the query bounds 
for security of ``Sum of the EM ciphers'' (SoEM), which are variants of the EM cipher \cite{CLM19}.
For example, they proved that $\OO(2^{n/2})$ queries are sufficient to classically attack $SoEM1(x):=\pi(x+k_1)+\pi(x+k_2)+k_1+k_2$ 
for two independent keys $k_1,k_2$ and $SoEM21(x):=\pi_1(x+k_1)+\pi_2(x+k_1)+k_1$ for two independent permutations $\pi_1,\pi_2$, but $\Omega(2^{2n/3})$ queries are necessary to classically attack
$SoEM22(x):=\pi_1(x+k_1)+\pi_2(x+k_2)+k_1+k_2$ for independent keys $k_1,k_2$ and independent permutations $\pi_1,\pi_2$ beyond the birthday bound.

Also, quantum attacks on the EM cipher and its variants have been developed following the result of Kuwakado and Morii. 
Shinagawa and Iwata demonstrated quantum attacks on the variants of SoEM studied in \cite{CLM19} 
by extending the Kuwakado-Morii (KM) attack \cite{SI22}. For example, they demonstrated 
that $SoEM1$ and $SoEM21$ can be broken only with $\OO(n)$ quantum queries. 
Moreover, their new quantum algorithm that combines Simon's algorithm with Grover's algorithm can break 
$SoEM22$ with $\OO(n2^{n/2})$ quantum queries, 
which is much lower than the classical query lower bound of $\Omega(2^{2n/3})$ \cite{CLM19}.
Bonnetain, Hosoyamada, Naya-Plasencia, Sasaki and Schrottenloher 
constructed a new quantum algorithm that uses Simon's algorithm as a subroutine 
without quantum queries to oracles, and they succeeded 
in attacking the EM cipher with $\OO(2^{n/3})$ classical queries, $\OO(n^2)$ qubits, 
and offline quantum computation $\tilde{\OO}(2^{n/3})$ \cite{BHN+19}. 

On the other hand, little has been studied on the security of these schemes 
against quantum adversaries, or limits of capabilities of quantum adversaries, 
while the KM attack has been used 
to extend quantum attacks on other variants of the EM cipher.
In many security proofs against quantum adversaries with oracle access, including the EM cipher and its variants, 
it is generally not possible to prove the security against quantum adversaries by conventional proof techniques 
used in the standard classical settings. 
This is because we need to assume that quantum adversaries have quantum access to cryptographic primitives. 
Indeed, many papers developed new techniques to show the limits of quantum adversaries 
against well-known symmetric-key cryptographic constructions (e.g., \cite{Zha12,Zha19}). 

The only example of the quantum security proof for the EM cipher, to the best of the authors' knowledge, 
is by Alagic, Bai, Katz, and Majenz \cite{ABKM21}. 
They considered a natural post-quantum scenario that adversaries make classical queries 
to its encryption function $EM$, but can make quantum queries to the public permutation $\pi$.
In this scenario, they demonstrated that it must hold either $q_\pi^2 q_{EM} = \Omega(2^n)$ 
or $q_\pi q_{EM}^2 = \Omega(2^n)$, where $q_\pi$ ($q_{EM}$, respectively) is the number 
of queries to $\pi$ ($EM$, respectively).

Therefore, it is important to understand better the limits of quantum adversaries 
for constructing quantumly secure variants of the EM cipher by studying 
the quantum query lower bounds for attacking the EM cipher.

In this paper, we investigate the limits of quantum adversaries 
against the original EM cipher to explore quantumly secure variants of the EM cipher.
We prove lower bounds $\Omega(n)$ of quantum query complexity to recover $n$-bit secret keys of the EM cipher 
even if quantum adversaries are allowed to make adaptive queries. 
To the best of the authors' knowledge, this is the first result that provides new techniques 
for demonstrating the limits of adversaries against (variants of) the EM cipher with purely quantum queries. 
Our quantum query lower bound matches the upper bound $\OO(n)$ of nonadaptive quantum queries 
provided by the KM attack up to a constant factor. 
This implies that their attack is optimal up to a constant factor in a setting of quantum query complexity, 
and thus, there is no asymptotically better quantum attack than the one based on Simon's algorithm 
even if it is allowed to make adaptive queries.

\section{Overviews of Previous Results and Our Ideas}\label{sec:overviews}
Since the structure of our proof is based on the optimality proof of (generalized) Simon's algorithm 
studied by Koiran, Nesme, and Portier \cite{KNP07}, we briefly review Simon's algorithm and its optimality.

The problem solved by Simon's algorithm is commonly referred to as Simon's problem. 
The following is a generalized version of Simon's problem with any prime $p$.
The oracle $O$ hides some subgroup $K$ of order $D=p^d$, where $d$ is a non-negative integer.

\begin{description}
 \magicwand
 \item[\rm \textbf{Generalized Simon's (GS) problem}]
 \item[\rm Input:] an oracle $O:\doubleZ_p^n\rightarrow Y$ that is sampled uniformly at random from all the oracles that satisfy 
 $
  x^\prime = x + k \leftrightarrow O(x^\prime)=O(x)
 $
 for some subgroup $K \le \doubleZ_p^n$ of order $D$;
 \item[\rm Output:] the generators of $K$.
\end{description}

The original Simon's problem corresponds to the case of $p=2$ and $D=2$. 
Then, $K=\set{0,k}$ for $k\in\doubleZ_2^n\setminus\set{0^n}$.
Simon's algorithm first makes $\OO(n)$ nonadaptive queries $\sum_{x\in\doubleZ_2^n}\ket{x}\ket{0}/\sqrt{2^n}$ 
to the oracle $O$ and measures the second register. 
By the measurement, it obtains independent copies of the coset-uniform state 
$(\ket{x_0}+\ket{x_0+k})/\sqrt{2}$ for a random $x_0$ in the first register. 
Applying the quantum Fourier transform over $\doubleZ_2^n$ 
(or, the Hadamard transform $H^{\otimes n}$) to them and measuring the resulting states, 
it obtains $\OO(n)$ random linear constraints 
$\sum_{i<n} z_i\cdot k_i =0$ with respect to the undetermined secret key $k=(k_0,\ldots,k_{n-1})$.
From the constraints, it can identify $k$ with constant probability.

The idea of the KM attack against the EM cipher is to construct 
the oracle of Simon's problem from the public permutation $\pi$ and encryption function 
$EM(x)=\pi(x+k_1)+k_2$. In the KM attack, a quantum adversary 
is allowed to make quantum queries to $\pi$ and $EM$ in a quantum manner. 
Let $O(x):=EM(x)+\pi(x)=\pi(x+ k_1)+\pi(x)+k_2$. 
The adversary applies Simon's algorithm to this function $O$.
Since $O(x+k_1)=\pi(x+k_1)+\pi(x)+k_2=O(x)$,  
The oracle $O$ satisfies the direct part $x^\prime = x + k_1 \rightarrow O(x^\prime)=O(x)$ 
and approximately satisfies the converse part with respect to random choices of $\pi$. 
Therefore, the KM attack succeeds in recovering $k_1$ using $\OO(n)$ nonadaptive quantum queries 
to $\pi$ and $EM$ with constant probability by Simon's algorithm. 
It is obvious to recover $k_2$ from $k_1$ since $k_2 = EM(0)+\pi(k_1)$.

To prove the optimality of (generalized) Simon's algorithm, Koiran et al.~studied 
quantum query lower bounds for 
its generalized decisional version.
Let 
\[
 F_D :=\set{O:\doubleZ_p^n\rightarrow Y: \exists K\le\doubleZ_p^n,
 \forall x^\prime,\forall x\in\doubleZ_p^n,\forall k\in K, 
x^\prime = x + k \leftrightarrow O(x^\prime) = O(x)},
\]
where $D=\setsize{K}=p^d$ for a non-negative integer $d$.
\begin{description}
 \magicwand
 \item[\rm\textbf{Generalized Decisional Simon's (GDS) problem}]
 \item[\rm Input:] an oracle $O:\doubleZ_p^n\rightarrow Y$ that is sampled uniformly at random from $F_p$ or $F_1$;
 \item[\rm Output:] ``accept'' if $O$ is from $F_p$ or ``reject'' if it is from $F_1$.
\end{description}
Note that $F_1$ is the set of all the $O:\doubleZ_p^n\rightarrow Y$, 
The task of this problem is to distinguish between a function 
that hides some subgroup $K$ of order $p$ and a random function.

It is easy to see that if Simon's problem is solved with $T$ queries, the GDS problem for $p=2$ is also solved with the same $T$ queries. 
Therefore, quantum query lower bounds of GDS problem for an arbitrary prime $p$ directly lead to those for Simon's problem. 

The argument of Koiran et al.~\cite{KNP07} is based on the polynomial method \cite{BBCMW01} for the GDS problem. 
They analyzed the degree of the polynomial $Q(D)$ that represents the accepting probability for a random $O\in F_D$, 
where $D$ is the order of the subgroup $K$ that $O$ hides.
They showed an upper bound $\OO(T)$ of $\deg(Q(D))$ for quantum algorithms with accepting probability $Q(D)$ and $T$ queries to an oracle $O$ 
that hides a subgroup $K$ of order $D$, and further, a lower bound $\Omega(n)$ of $\deg(Q(D))$ 
for any polynomial $Q(D)$ that satisfies several conditions naturally posed on $Q(D)$, 
such as $Q(p)\ge 1-\epsilon$, which corresponds to the case of $F_p$, $Q(1)\le \epsilon$, 
which corresponds to the case of $F_1$, for a small constant $\epsilon$, 
and $Q(p^i)\in [0,1]$ for every $i\in\set{0,1,\ldots, n}$.

Our goal, quantum query lower bounds for key recovery of the EM cipher, seems to be close to 
those for Simon's problem provided in \cite{KNP07}. 
However, there are actually technical gaps between these two problems. 
In the setting of the key recovery, a quantum adversary can make access to two oracles $EM(x)$ and $\pi(x)$ 
rather than a single oracle $O(x)$ in the setting of Simon's problem. The quantum query upper bound $\OO(n)$ 
can be achieved by the KM attack that synchronously makes a (quantumly superposed) query $x$ to $EM(x)$ and $\pi(x)$ 
and combines two answers to compute $O(x)=EM(x)+\pi(x)$. 
However, it would be possible to achieve better attacks by making different queries to two oracles in an adaptive manner.

We then provide a reduction of quantum query lower bounds in the standard query model to those in a special query model. 
In the special query model, which we refer to as a synchronized query model, any quantum adversary is posed to make a synchronized query 
to two oracles as done in the KM attack. If a quantum adversary $A$ can recover the secret key 
with $T(n)$ queries to $EM$ and $\pi$ totally in the standard query model, we can easily modify $A$ 
to another adversary $A^\prime$ that recovers it with $2T(n)$ queries in the synchronized query model. 

In the synchronized query model, we can assume that a quantum algorithm has synchronized access to 
a oracle sequence $O:\doubleZ_2^n\rightarrow(\doubleZ_2^n)^2$, where $O(x)=(O_0(x),O_1(x))$ for a random permutation $O_0(x)=\pi(x)$ 
and the encryption function $O_1(x)=EM(x)=O_0(x+k_1)+k_2$.
In our proof, we focus only on the inner key $k_1$ for simplification, which suffices to prove lower bounds 
since it is a special case when $k_2=0$. We define $O_1(x)=O_0(x+k_1)$. Then, our goal is to prove quantum query lower bounds 
for finding the inner key $k_1$ with synchronized queries to the oracle sequence $O(x)=(O_0(x),O_1(x))=(O_0(x),O_0(x+k_1))$.

To apply the polynomial method as done in the proof of Koiran et al., 
we need to consider a generalized version of the oracle sequence 
$O(x)=(O_0(x),O_1(x))$ to represent the accepting probability as a polynomial 
in some single parameter.

As a generalization, we consider an oracle sequence 
\begin{align*}
 O(x) & =(O_0(x),O_1(x),\ldots,O_{D-1}(x))\\
 & =(O_0(x+k_0), O_0(x+k_1), \ldots, O_{0}(x+k_{D-1}))
\end{align*}
of length $D=p^d$, where $K=\set{k_0=0^n, k_1, \ldots, k_{D-1}}$ is a subgroup in $\doubleZ_p^n$ 
of the order $D$. We then analyze the accepting probability $Q(D)$ as a polynomial in $D$ 
for a given oracle sequence $O$. 

The major difference from the argument of Koiran et al.~is an algebraic structure behind the oracles.
In the cases of the GS and GDS problems, the subgroup is hidden 
in the single oracle. However, it is hidden in the correlation among $D$ oracles in our setting.
Recall that $x^\prime = x + k$ for some $k\in K$ 
if and only if $O(x^\prime)=O(x)$ in Simon's problem.
We need to reveal a similar algebraic structure to analyze of the degree of $Q(D)$.

Our idea is to characterize the order of oracles in the sequence $O$ by the hidden subgroup $K$.
Actually, we demonstrate that the definition of $O$ is equivalent with the statement that 
$x^\prime = x + k_i$ for $k_i\in K$ if and only if 
for $k_i\in K$ there exists some permutation $\sigma_i$ over $\set{0,\ldots,D-1}$ 
it holds $O(x^\prime)=\sigma_i O(x)$.

Let us consider a small example $K=\set{0^n,k_1,k_2,k_1+k_2} \le \doubleZ_2^n$ 
for $p=2$ and $D=4$, where $k_1\ne k_2\in\doubleZ_2^n\setminus\set{0^n}$.
The oracle sequence is defined as 
\begin{align*}
O(x) & =(O_0(x),O_0(x+k_1),O_0(x+k_2),O_0(x+k_1+k_2))\\
& =(O_{(0,0)}(x), O_{(0,1)}(x), O_{(1,0)}(x), O_{(1,1)}(x))
\end{align*}
with some special indexing of the oracles.
Then, we can see that 
\begin{align*}
O(x+k_1) & =(O_0(x+k_1), O_0(x), O_0(x+k_1+k_2), O_0(x+k_2))\\
 & =(O_{(0,0)+(0,1)}(x), O_{(0,1)+(0,1)}(x), O_{(1,0)+(0,1)}(x), O_{(1,1)+(0,1)}(x))\\
 & =(O_{(0,1)}(x), O_{(0,0)}(x), O_{(1,1)}(x), O_{(1,0)}(x)).
\end{align*}
Similarly, we have
\begin{align*}
O(x+k_2) & =(O_{(1,0)}(x), O_{(1,1)}(x), O_{(0,0)}(x), O_{(0,1)}(x))\\
O(x+k_1+k_2) & =(O_{(1,1)}(x), O_{(1,0)}(x), O_{(0,1)}(x), O_{(0,0)}(x)).
\end{align*}
Hence, every $k\in K$ corresponds to some permutation over the order of the oracles.

From the above characterization, we develop a variant of the argument 
of Koiran et al.~based on the polynomial method with the analogous property 
of the oracle sequence that $O(x+k_i)=\sigma_i O(x)$ instead of the one of 
Simon's problem that $O(x+k)=O(x)$.
As is obvious, the analogous property is different from that of Simon's problem, 
and hence, we need to fill this gap with other technical tricks in our proof.

\section{Preliminaries}\label{sec:prel}
Before describing the main result, we briefly discuss the formal treatment of quantum query algorithms.

In the context of quantum query complexity, we usually assume the following framework for quantum query algorithms. 
A quantum algorithm $A^O$ with a given oracle $O$ has quantum memory of three registers $\ket{x}\ket{y}\ket{z}$, 
where the first one is the query register which stores a query to $O$, the second one is the answer register which stores an answer from $O$, 
and the third one is the working register which stores all the other than the query and answer registers. 
Let $U_O$ be the oracle gate of $O:X \rightarrow Y$ that acts on the query and answer registers: 
$U_O\ket{x}\ket{y}=\ket{x}\ket{O(x)\oplus y}$ for every $x\in X$ and every $y\in Y$. $A$ starts with the initial state $\ket{0}\ket{0}\ket{0}$, 
and applies an arbitrary unitary operator to all the three registers and then applies $U_O$ to the two registers alternatively. 
Then, the $A^O$'s final state is provided as $\ket{\psi_T}=U_{T} (U_O\otimes I) U_{T-1} \cdots U_1 (U_O\otimes I) U_0\ket{0}\ket{0}\ket{0}$.

The $A^O$'s output can be obtained by measuring a part of the final state in the computational basis.
Note that this formulation allows $A$ to make adaptive queries. In other words, 
$A$ can make a query that depends on the answers to the previous queries.

In this paper, we need to deal with multiple oracles such as $\pi$ and $EM$. We formulate the quantum query model 
with multiple oracles $O_0,O_1,\ldots,O_{N-1}$ by the model with a single oracle $O:\set{0,1,\ldots,N-1}\times X\rightarrow Y$ 
defined as $O(i,x):=O_i(x)$. In the framework for quantum query algorithms, this oracle can be implemented 
as $U_O\ket{i,x}\ket{y}\ket{z}=\ket{i,x}\ket{O_i(x)\oplus y}\ket{z}$ by extending the query register.

As described in Section~\ref{sec:overviews}, we also consider a special query model referred to as the synchronized query model. 
A quantum query algorithm $A$ receives $N$ answers $O_0(x),\ldots,O_{N-1}(x)$ simultaneously on a single query $x$ at its oracle call 
in the synchronized query model. Formally, the oracle call can be implemented as 
$U_{O}\ket{x}\ket{y_0,\ldots,y_{N-1}}=\ket{x}\ket{O_0(x)\oplus y_0,\ldots,O_{N-1}(x)\oplus y_{N-1}}$. 
Similarly to the standard query model, $A$ applies an arbitrary unitary operator to the registers, and then, the oracle operator $U_O$, 
with the all-zero initial state. We count the number of queries as the number of $U_O$ used in the algorithm. 
We also regard the oracle $O$ as a function $O: X\rightarrow Y^N$ by setting $O(x):=(O_0(x),\ldots,O_{N-1}(x))$ in this model.

As mentioned in Section \ref{sec:intro}, any quantum algorithm in the standard query model can be converted to the one in the synchronized query model from the following proposition.
The proof is easily done by a standard reduction.

\begin{proposition}\label{prop:simulation}
Let $A$ be any quantum query algorithm with $T$ queries in the standard query model.
Then, there exists $A^\prime$ with $2T$ queries in the synchronized query model such that 
$A^\prime$'s output distribution is identical with $A$'s one.
\end{proposition}
\begin{proof}
Let $\ket{\psi_t}:=\sum_{i,x,y,z}\alpha_{i,x,y,z}\ket{i,x}\ket{y}\ket{z}$ be the quantum state of $A$ before the $t$-th query. 
At the $t$-th oracle call, the state changes to 
$\sum_{i,x,y,z}\alpha_{i,x,y,z}\ket{i,x}\ket{O_i(x)\oplus y}\ket{z}$
in the standard query model. Then, $A$ applies $U_t$ to this state to obtain $\ket{\psi_{t+1}}$.

We simulate the change from $\ket{\psi_t}$ to $\ket{\psi_{t+1}}$ with $2$ queries in the synchronized query model.
For induction, let us assume that we obtain 
$\ket{\psi_{2t}^\prime}:=\sum_{i,x,y,z}\alpha_{i,x,y,z}\ket{i,x}\ket{0}\cdots\ket{0}\ket{y}\ket{z}$ 
that has all-zero $k$ answer registers before the $2t$-th query in the synchronized query model. 
(The simulation of the base state $U_0\ket{0}\ket{0}\ket{0}$ is trivial.)

Starting from $\ket{\psi_{2t}^\prime}$,  we first apply $U_O$ to the $k$ answer registers on a query $x$. 
We next xor the $i$-th answer register $\ket{O_i(x)}$ into $\ket{y}$. The state changes to 
$\sum_{i,x,y,z}\alpha_{i,x,y,z}\ket{i,x}\ket{O_0(x),\ldots,O_{k-1}(x)}\ket{O_i(x)\oplus y}\ket{z}$.
We then apply $U_O$ again to clean the first $k$ answer registers. 
The state changes to $\sum_{i,x,y,z}\alpha_{i,x,y,z}\ket{i,x}\ket{0}\cdots\ket{0}\ket{O_i(x)\oplus y}\ket{z}$. 
By applying $U_t$ to $\ket{i,x}\ket{O_i(x)\oplus y}\ket{z}$, we obtain the state $\ket{\psi_{2t+2}^\prime}$ 
that simulates $\ket{\psi_{t+1}}$.
\end{proof}

From Proposition~\ref{prop:simulation}, if we obtain a query lower bound of $T$ in the synchronized query model, 
we also obtain a query lower bound of $T/2$ in the standard query model. 
Thus, we focus on the synchronized query model in the remaining part of this paper.

We next discuss our target problem to prove the quantum query lower bounds for the key recovery of the EM cipher.
As done in \cite{KNP07}, we work on a decisional version of attacks against the EM cipher. 
In the key recovery problem for the EM cipher, we need to deal with multiple oracles 
such as $\pi$ and $EM$, unlike the GDS problem. We are given two oracles $O_0:=\pi$ and $O_1:=EM$, 
where $O_0:\doubleZ_2^n\rightarrow\doubleZ_2^n$ is a public permutation 
and $O_1(x)=\pi(x\oplus k_1)\oplus k_2$ for secret keys $k_1,k_2\in\doubleZ_2^n$. Then, the task is to recover $k_1,k_2$ via queries to $O_0$ and $O_1$. 
We focus on a special case $k_2=0^n$ of the key recovery problem since a lower bound for this special case implies that for the general case. 

To apply the polynomial method similarly to \cite{KNP07}, we consider a generalized version of 
the key recovery problem. One of the main technical contributions is a formalization of 
the generalized version, named generalized decisional 
inner-key only EM cipher (GDIKEM) problem, that is suitable for proving query lower bounds.

Note that query lower bounds of the key recovery problem in the standard query model 
can be obtained from the GDIKEM problem in the query synchronized model by Proposition~\ref{prop:simulation}.
Therefore, we can suppose that 
a quantum query algorithm is provided an oracle sequence $O(x)=(O_0(x),\ldots, O_{N-1}(x))$ 
in the definition of the GDIKEM problem rather than a set of oracles $O_0,\ldots, O_{N-1}$ separately. 

Before the definition of the GDIKEM problem, we consider a special index system 
$I=\set{(i_0,\ldots,i_{d-1}): i_0,\ldots,i_{d-1}\in\doubleZ_p}$ 
for the oracle sequences $O$. 
Let $K$ be any subgroup of $\doubleZ_p^n$ of order $D=p^d$.
We fix the lexicographic first set $\set{g_0^K,\ldots,g_{d-1}^K}$ 
of generators for $K$. Then, any element $k_i\in K$ can be associated 
with $i\in I$ to satisfy $k_i:=\sum_{j=0}^{d-1} i_j g_j^K$.
Note that $k_i+k_{i^\prime} = k_{i+{i^\prime}}$ for $k_i,k_{i^\prime}\in K$.
For simplification, let $0$ denote $0^d$. We sometimes identify $I$ 
with $\set{0,1,\ldots, D-1}$ by the lexicographical order.

To formulate the GDIKEM problem, we define a set of oracle sequences 
of length $D$ as $O(x)=(O_i(x))_{i\in I}$, where 
$O_i:\doubleZ_p^n\rightarrow\doubleZ_p^n$ is a permutation.
Let 
\begin{align*}
 F_D :=  \bigg\{O: \exists K\le\doubleZ_p^n\,(\setsize{K}=D),
 \forall x\in\doubleZ_p^n, \forall i\in I, O_i(x)=O_0(x+k_i)\bigg\},
\end{align*}
where $D=p^d$ for some $d$.
For $O\in F_D$, we say that $O$ hides a subgroup $K$. 

Note that $F_2$ is a set of the oracles $O(x)=(O_0(x),O_1(x))=(O_0(x+0^n),O_0(x+k_1))$ 
for a subgroup $K=\set{0^n,k_1}$ in the case when $D=p=2$, 
which 
corresponds to instances of the EM cipher 
only with an inner key $k_1$ and public random permutation $O_0$.

From the following reason, we can see that every $O\in F_D$ hides the unique subgroup $K$ of order $D$.
Assume that $O$ hides two distinct subgroups $K$ and $K'$ of order $D$. 
For $k'\in K'\setminus K$, there exists some index $i$ $O_i(x)=O_0(x+k')$.
Then, some $k\in K$ is associated with the index $i$, and thus, $O_i(x)=O_0(x+k)$.
Hence, $O_0(x+k')=O_0(x+k)$. However, since $x+k\neq x+k'$, $O_0$ cannot be a permutation.
This is a contradiction. Therefore, a subgroup hidden by $O$ is unique.

By analogy with the GDS problem, it would be natural to define the distinguishing task 
between oracle sequences from $F_p$ and $F_1$. However, these oracle sequences from  
$F_p$ and $F_1$ are of different output lengths. To align the lengths, 
we pad redundant oracles to them.
We define a set $\hat{F}_{D,N}$ of oracle sequences of length $N$ 
$\hat{O}=(O_0,\ldots,O_{N-1})$ such that $(O_0,\ldots,O_{D-1})\in F_D$ and 
$O_i$ is an arbitrary permutation over $\doubleZ_p^n$ for $i\ge D$. 

Now, we define the GDIKEM problem as follows.
\begin{description}
 \item[\rm \textbf{GDIKEM problem}]
 \item[\rm Input:] an oracle $\hat{O}$ that satisfies ($i$) $\hat{O}\in\hat{F}_{D,N}$ or ($ii$) $\hat{O}\in\hat{F}_{1,N}$.
 \item[\rm Output:] ``accept'' if ($i$) or ``reject'' if ($ii$).
\end{description}

$F_1$ contains all the permutations over $\doubleZ_p^n$, and hence, 
$\hat{F}_{1,N}$ is the set of all the possible sequences permutations 
over $\doubleZ_p^n$ of length $N$.
On the other hand, $F_2$ contains pairs of the permutations 
$(O_0(x),O_0(x+k_1))$ for some subgroup $K=\set{0^n,k_1}$ of order $2$ 
in the case when $D=p=2$.
Therefore, $\hat{F}_{1,N}$ and $\hat{F}_{2,N}$ correspond 
to the sets of accepting and rejecting instances of a decisional version 
(with redundant $N-2$ padded oracles) of the attack against EM cipher,
respectively.

In this paper, we show that every quantum algorithm $A$ requires $\Omega(n)$ queries 
if $A^{\hat{O}}$ accepts for a randomly chosen oracle $\hat{O}$ in the case ($i$) with at most $\epsilon$ 
and for a randomly chosen oracle in the case ($ii$) with least $1-\epsilon$, where $\epsilon$ is a fixed constant.
If there exists a key-recovery quantum algorithm for permutations $O_0(x)$ and $O_1(x)=O_0(x+k_1)$ 
with some $k_1\ne 0^n$, it also works for the GDIKEM problem. Thus, query lower bounds of the GDIKEM problem 
imply those of the key recovery.

\section{Proof of Quantum Query Lower Bounds}\label{sec:LB}
We demonstrate our main result, quantum query lower bounds for key recovery attacks against the EM cipher, in this section.

As used in the previous result of Koiran et al.~\cite{KNP07}, we characterize the acceptance probability 
of any quantum algorithm for the oracle $O$ from a set of partial functions whose domain size by the number of queries using the polynomial method \cite{BBCMW01}. 

We say $f$ extends $s$, which is also denoted by $f\supseteq s$, if $s(x)=f(x)$ for every $x\in\dom(s)$. 
For any function $f:X\rightarrow Y^N$ and any partial function $s:X\rightarrow Y^N$, we define 
\[
 I_s(f):=\begin{cases} 1 & \text{if $f$ extends $s$}; \\ 0 & \text{otherwise.} \end{cases}
=\prod_{ \substack{x\in\dom(s),\\ s(x)=\bar{y}} }\Delta_{x,\bar{y}}(f),
\]
where $\Delta_{x,\bar{y}}(f)=1$ if $f(x)=\bar{y}$ and $\Delta_{x,\bar{y}}(f)=0$ otherwise.

Similarly to \cite{KNP07}, we can prove the following characterization (Theorem~\ref{thm:poly}) of the acceptance probability 
with respect to $I_s(f)$ even in the synchronized query model. 
The proof follows from the same argument as the one of the standard polynomial method. 

\begin{theorem}\label{thm:poly}
Let $A$ be any quantum algorithm with $T$ queries in the synchronized query model. 
Then, there exists a set $S$ of partial functions $s:X\rightarrow Y^N$ such that $A$ accepts $f$ with probability  
 $P(f):=\sum_{s\in S} c_s I_s(f)$
for some real numbers $c_s$, where $\setsize{\dom(s)}\le 2T$.
\end{theorem}
\begin{proof}
Let $\ket{\psi^{(t)}}$ be the $A$'s state after $t$ queries. 
Namely, $\ket{\psi^{(0)}} := U^{(0)}\ket{0}\ket{0}^{\otimes m}\ket{0}$ 
and $\ket{\psi^{(t)}} := U^{(t)}(U_f\otimes I) U^{(t-1)} \cdots U^{(1)}(U_f\otimes I)\ket{\psi^{(0)}}$.
Let
\[
 \ket{\psi^{(t)}} := \sum_{x,\bar{y},z}\alpha^{(t)}_{x,\bar{y},z}\ket{x}\ket{\bar{y}}\ket{z}.
\]
For induction, we assume that there exists a set $S^{(t)}$ of partial functions $s$ 
with $\setsize{\dom(s)}\le t$ such that for some complex-valued coefficients $c^{(t)}_{x,\bar{y},z,s}$ 
\[
 \alpha^{(t)}_{x,\bar{y},z} = \sum_{s\in S^{(t)}} c^{(t)}_{x,\bar{y},z,s}I_s(f).
\]

By applying $U_f$ to $\ket{\psi^{(t)}}$, we obtain
\begin{align*}
 (U_f\otimes I)\ket{\psi^{(t)}} 
 & = \sum_{x,\bar{y},z,\bar{w}} \alpha^{(t)}_{x,\bar{y},z}\ket{x}\ket{\bar{y}\oplus f(x)}\ket{z}\\
 & = \sum_{x,\bar{y},z,\bar{w}} \Delta_{x,\bar{w}}(f)\alpha^{(t)}_{x,\bar{y}\oplus\bar{w},z}\ket{x}\ket{\bar{y}}\ket{z}\\
 & = \sum_{x,\bar{y},z,\bar{w}} 
 \sum_{s\in S^{(t)}} c^{(t)}_{x,\bar{y}\oplus\bar{w},z,s} \Delta_{x,\bar{w}}(f)I_s(f)\ket{x}\ket{\bar{y}}\ket{z}.
\end{align*}

We define a partial function $s^\prime$ for every $s\in S^{(t)}$
as $s^\prime(v):=s(v)$ for $v\in\dom(s)$ and $s^\prime(x):=\bar{w}$
if $x\notin\dom(s)$. Then, we obtain a set $S^{(t+1)}$ of partial functions $s^\prime$ 
of domain size at most $t+1$,
and then, for these $s\in S^{(t)}$ and $s^\prime\in S^{(t+1)}$ we have $I_{s^\prime}(f)=\Delta_{x,\bar{w}}(f)I_s(f)$.

Furthermore, a unitary operator $U^{(t+1)}$ only yields a linear combination of the coefficients $c^{(t)}_{x,\bar{y}\oplus\bar{w},z,s} \Delta_{x,\bar{w}}(f)I_s(f)$,
and hence, we obtain
$\ket{\psi^{(t+1)}}=U_{t+1}(U_f\otimes I)\ket{\psi^{(t)}}
 =\sum_{x,\bar{y},z}\alpha^{(t+1)}_{x,\bar{y},z}\ket{x}\ket{\bar{y}}\ket{z}$, 
 where $\alpha^{(t+1)}_{x,\bar{y},z} = \sum_{s\in S^{(t+1)}} c^{(t+1)}_{x,\bar{y},z,s}I_s(f)$ 
 for the set $S^{(t+1)}$ of partial functions $s$ with $\setsize{\dom(s)}\le t+1$.

By induction, we obtain  $\alpha^{(T)}_{x,\bar{y},z} = \sum_{s\in S^{(T)}} c^{(T)}_{x,\bar{y},z,s}I_s(f)$ 
for the set $S^{(T)}$ of partial functions $s$ with $\setsize{\dom(s)}\le T$ at the final state $\ket{\psi^{(T)}}$.

Let $G$ be a set of the acceptance bases of $A$. Then, the acceptance probability $P(f)$ is provided by 
\begin{align*}
 P(f) & = \sum_{(x,\bar{y},z)\in G} \lvert \alpha^{(T)}_{x,\bar{y},z}\rvert^2
  = \sum_{s,s^\prime\in S^{(T)}}\sum_{(x,\bar{y},z)\in G} (c^{(T)}_{x,\bar{y},z,s})^*c^{(T)}_{x,\bar{y},z,s^\prime} I_s(f)I_{s^\prime}(f).
\end{align*}
We define a partial function $s^{\prime\prime}$ as $s^{\prime\prime}(x)=s(x)$ for $x\in\dom(s)$ and for $s^{\prime\prime}(x)=s^\prime(x)$
$x^\prime\in\dom(s^\prime)$. Then, we obtain a set $S$ of partial functions $s^{\prime\prime}$ with $\setsize{\dom(s^{\prime\prime})}\le 2T$.
(For $x\in\dom(s)\cap\dom(s^\prime)$, we define $s^{\prime\prime}(x)=s(x)$. 
Even if $s(x)\ne s^{\prime}(x)$ for this $x$, $I_s(f)I_{s^\prime}(f)=0$ for such $s,s^\prime$, and hence, the definition is consistent.)
For the set $S$, it holds 
$
 P(f)=\sum_{s\in S} c_s I_s(f)
$
for some real coefficients $c_s$.
\end{proof}

As stated in Section~\ref{sec:intro}, we focus on the degree of a polynomial 
that represents accepting probability of a quantum algorithm to prove 
the query lower bounds by the polynomial method.

In Section~\ref{sec:prel}, we defined the GDIKEM problem to 
naturally fit some generalized decisional version of the attack against the EM cipher. 
From technical reasons, we focus on another equivalent 
formulation of the oracle set shown in the following lemma.

\begin{lemma}\label{lem:equiv}
Suppose that $O$ hides a subgroup $K$.
Then, we have
\begin{align*}
  F_D & = \bigg\{O: \exists K\le\doubleZ_p^n\ (\setsize{K}=D)\\
  & \qquad\qquad \forall i\in I,\forall x,\forall x^\prime\in \doubleZ_p^n,\ x^\prime=x+k_i\ (k_i\in K) \leftrightarrow O(x^\prime)=\sigma_i O(x)\bigg\}.
\end{align*}
\end{lemma}
\begin{proof}
Suppose that $O_i(x)=O_0(x+k_i)$ for every $x\in\doubleZ_p^n$ and every $i\in I$, where $k_i\in K$.
Since $O_j(x+k_i)=O_0(x+k_i+k_j)=O_0(x+k_{i+j})=O_{i+j}(x)$,
\begin{align*}
 O(x+k_i) = (O_j(x+k_i))_{j\in I} = (O_{i+j}(x))_{j\in I} = \sigma_i O(x).
\end{align*}
Thus, $x^\prime=x+k_i$ implies $O(x^\prime)=\sigma_i O(x)$.

If $O(x^\prime)=\sigma_i O(x)$, we have
\begin{align*}
 O(x^\prime) = \sigma_i O(x) = (O_{i+j}(x))_{j\in I} = (O_j(x+k_i))_{j\in I}.
\end{align*}
Hence, we have $O_0(x+k_i)=O_0(x^\prime)$. Since $O_0$ is a permutation, 
$x^\prime=x+k_i$ holds. Therefore, $O(x^\prime)=\sigma_i O(x)$ implies $x^\prime=x+k_i$.
\end{proof}

From technical reasons, we define a subset $F_D^*:= F_D \cap \set{O: O_0\in\Pi_K}$ of the oracles.
The set $\Pi_K$ of permutations is defined as follows. 
Let $K$ be the subgroup hidden by $O$.
We consider the coset decomposition of $\doubleZ_p^n$ for $K$: 
$\doubleZ_p^n = \cup_{i<N/D}\set{c_i+K}$ for some fixed representatives, 
where $c_0:=0^n$ and $N:=\setsize{\doubleZ_p^n}=p^n$.
To construct $\Pi_K$, for every sequence $(a_0,\ldots,a_{(N/D)-1})$ 
of distinct $N/D$ elements, we put a permutation $\pi$ into $\Pi_K$ 
such that $\pi(c_0)=a_0,\ldots, \pi(c_{(N/D)-1})=a_{(N/D)-1}$ 
and the remaining values $\pi(x)$ for $x\notin{c_0,\ldots,c_{(N/D)-1}}$ 
are determined by the lexicographically first sequence of $N-(N/D)$ elements 
excluding $a_0,\ldots,a_{(N/D)-1}$ from $\doubleZ_p^n$.
Therefore, any permutation in $\Pi_K$ is determined uniquely 
by specifying the values $\pi(c_0),\ldots,\pi(c_{(N/D)-1})$,
and thus, $\setsize{\Pi_K}=p^n(p^n-1)\cdots(p^n-(p^{n-d}-1))$.
We also define its padded version $\hat{F}_{D,N}^*$ by the same manner as $\hat{F}_{D,N}$.

We now provide a formal statement of our main theorem.
\begin{theorem}\label{thm:main}
Let $p$ be any prime, and let $\epsilon$ be any constant in $(0,1/2)$. 
Suppose that $A$ is any quantum algorithm with adaptive $T=T(n)$ quantum queries 
to a given oracle $\hat{O}:\doubleZ_p^n\rightarrow(\doubleZ_p^n)^N$, 
where $\hat{O}$ is sampled uniformly from ($i$) $\hat{F}_{p,N}^*$ 
or ($ii$) $\hat{F}_{1,N}^*$ for any fixed $N\ge p$. 
If $A^{\hat{O}}$ accepts with at least $1-\epsilon$ in the case ($i$) and 
with at most $\epsilon$ in the case ($ii$), it holds that $T=\Omega(n)$.
\end{theorem}

Immediately from Proposition~\ref{prop:simulation} and Theorem~\ref{thm:main}, we obtain a quantum query 
lower bound of $\Omega(n)$ to recover secret keys in the EM cipher 
with constant success probability in the standard query model.

\begin{proofof}{Theorem~\ref{thm:main}}
We analyze the accepting probability that $A^{\hat{O}}$ accepts for an oracle $\hat{O}\in\hat{F}_{D,N}^*$.
From Theorem~\ref{thm:poly}, the accepting probability is
\begin{align*}
  P(\hat{O}) = \sum_{\hat{O}\in\hat{F}_{D,N}^*}\sum_{\hat{s}\in \hat{S}}c_{\hat{s}} I_{\hat{s}}(\hat{O})
 = \sum_{\hat{O}\in\hat{F}_{D,N}^*}\sum_{\hat{s}\in \hat{S}}c_{\hat{s}} \prod_{x\in\dom{\hat{S}},\bar{y}=\hat{s}(x)}\Delta_{x,\bar{y}}(\hat{O})
\end{align*}
for some set $\hat{S}$ of partial functions.

We convert this multivariate polynomial $P(\hat{O})$ in $\set{\Delta_{x,\bar{y}}(\hat{O})}_{x,\bar{y}}$ into 
another univariate polynomial $Q(D)$ in $D$ by averaging the redundant oracles, namely,
\[
 Q(D) := \frac{1}{\setsize{\hat{F}_{D,N}^*}}\sum_{\hat{O}\in\hat{F}_{D,N}^*}P(\hat{O}).
\]
Recall that $\hat{O}$ is padded with $N-D$ redundant oracles to align the length of the oracle sequences. 
From the following lemma (Lemma~\ref{lem:redundant}), we can ignore such redundant oracles for the degree analysis of $Q(D)$. 
\begin{lemma}\label{lem:redundant}
There exists a set of partial functions $S$ such that for every $O\in F_D$ we have  
\[
 Q(D) = \frac{1}{\setsize{F_D^*}}\sum_{O\in F_D^*}\sum_{s\in S} c_s^\prime I_s(O)
\]
\end{lemma}
\begin{proof}
Let $O^\prime=(O_{D},\ldots,O_{N-1})$ be a uniformly random oracle sequence in $\hat{O}$. 
From the definition, we have 
\[
 Q(D) = \frac{1}{\setsize{F_D^*}} \sum_{O\in F_D^*} \Expect{O^\prime\in F^{N-D}}{P(\hat{O})}.
\]

Averaging the redundant oracles in $P(\hat{O})$ over $O^\prime$,
\begin{align*}
 & \Expect{O^\prime\in F^{N-D}}{P(\hat{O})}
  = \Expect{O^\prime\in F^{N-D}}{\sum_{\hat{s}\in\hat{S}}c_{\hat{s}}\prod_{ \substack{x\in\dom(\hat{s})\\ \bar{y}=\hat{s}(x)}}\Delta_{x,\bar{y}}(\hat{O})}\\ 
  & = \Expect{O^\prime\in F^{N-D}}{\sum_{(s,s^\prime)\in\hat{S}}c_{(s,s^\prime)}\prod_{ \substack{x\in\dom(s^\prime)\\ \bar{y}^\prime=s^\prime(x)}}\Delta_{x,\bar{y}^\prime}(O^\prime)
 \prod_{\substack{x\in\dom(s)\\ \bar{y}=s(x)}}\Delta_{x,\bar{y}}(O)
 }\\
 & = \sum_{s\in S} \Expect{O^\prime\in F^{N-D}}{\sum_{s^\prime\in S^\prime_s} c_{(s,s^\prime)}\prod_{ \substack{x\in\dom(s^\prime)\\ \bar{y}^\prime=s^\prime(x)}}\Delta_{x,\bar{y}^\prime}(O^\prime)}
 \prod_{\substack{x\in\dom(s)\\ \bar{y}=s(x)}}\Delta_{x,\bar{y}}(O),
\end{align*}
where $S:=\set{s: \exists s^\prime\ (s,s^\prime)\in\hat{S}}$ and 
$S^\prime_s:=\set{s^\prime: (s,s^\prime)\in\hat{S}}$.
Note that $\dom(\hat{s})=\dom(s)=\dom(s^\prime)$ in the synchronized query model.

Let $F$ be the set of all the permutations over $\doubleZ_p^n$ and 
let $\delta_{s^\prime}$ be a fraction of permutations in $F$ that extends $s^\prime$.
Then, we have 
\begin{align*}
  & \Expect{O^\prime\in F^{N-D}}{\sum_{s^\prime\in S^\prime_s} c_{(s,s^\prime)}\prod_{ \substack{x\in\dom(s^\prime) \bar{y}^\prime=s^\prime(x)}}\Delta_{x,\bar{y}^\prime}(O^\prime)}\\
  & = \frac{1}{\setsize{F}^{N-D}}\sum_{O^\prime\in F^{N-D}}\sum_{s^\prime\in S^\prime_s} c_{(s,s^\prime)}\prod_{ \substack{x\in\dom(s^\prime)\\ \bar{y}^\prime=s^\prime(x)}}\Delta_{x,\bar{y}^\prime}(O^\prime)
  = \sum_{s^\prime\in S^\prime_s} c_{(s,s^\prime)}\delta_{s^\prime}.
\end{align*}

Let $c^\prime_s := \sum_{s^\prime\in S^\prime_s} c_{(s,s^\prime)}\delta_{s^\prime}$.
We have 
\begin{align*}
  \Expect{O^\prime\in F^{N-D}}{P(\hat{O})} = \sum_{s \in S} c^\prime_s I_s(O).
\end{align*}
\end{proof}

The following lemma shows $\deg(Q(D))$ is upper-bounded by the domain size of partial functions $s$.
\begin{lemma}\label{lem:domain-size}
Let $A$ be any quantum algorithm with $T$ queries in the synchronized query model. 
Then, we have $\deg(Q(D))\le\max_{s\in S}\setsize{\dom(s)}$.
\end{lemma}

By combining Theorem~\ref{thm:poly} and Lemma~\ref{lem:domain-size}, the lower bound of $T$ can be reduced to that of the degree of $Q(D)$. 
Koiran et al.~provided the degree analysis in \cite{KNP07}, which we apply in our proof.
\begin{theorem}[Koiran et al.~\cite{KNP07}]\label{thm:degree}
Let $c>0$ and $\xi>1$ be constants and let $P$ be a real polynomial with following properties:
$(i)$~$\lvert P(\xi^i)\rvert \leq 1$, for any integer $0 \leq i \leq n$, and 
$(ii)$~$\lvert dP(x_0)/dx \rvert \geq c$, for some real number $1 \leq x_0 \leq \xi$.
Then
\[
\mathrm{deg}(P) \geq \min \left\{n/2, \left(\log_2\left(\xi^{n+3}c\right)-1\right)/\left(\log_2\left(\frac{\xi^3}{\xi-1}\right)+1\right) \right\}.
\]
\end{theorem}

Let $A$ be any quantum algorithm solving GDIKEM problem for $\setsize{K}=p$
with bounded error probability $\epsilon$ and $T$ queries in the synchronized query model.
$A^O$ rejects if $\setsize{K}=1$ holds in GDIKEM problem, and $A^O$ accepts if $\setsize{K}= p$.
Then, $0 \leq \lvert Q(p^i)\rvert \leq 1 \; (0 \leq i \leq n)$ and $Q(p) \geq 1-\epsilon\;(k \leq n), \; Q(1) \leq \epsilon$ holds from the property of $A$.
Therefore, for the derivative of the polynomial $Q$, $Q$ satisfies $\lvert dQ(x_0)/dD\rvert \geq \frac{1-2\epsilon}{p-1}$ for some $x_0 \;(1 \leq x_0 \leq p)$ and $Q(p^i)\in [0,1]$ for any $i \in \set{0,...,n}$.
By applying Theorem \ref{thm:degree} to the polynomial $P=2Q-1$, we obtain the following inequality 
\[
    \deg(Q) \geq \min \left\{ n/2, \left(\log_2\left(\frac{p^{n+3}}{p-1}(2-4\epsilon)\right)-1\right)/\left(\log_2\left(\frac{p^3}{p-1}\right)+1\right) \right\} = \Omega(n).
\]

Therefore, the remaining task for the proof of the lower bound is to show Lemma~\ref{lem:domain-size}.
\begin{proofof}{Lemma~\ref{lem:domain-size}}
From Lemma~\ref{lem:redundant}, we have
\begin{align*}
Q(D) = & \frac{1}{\setsize{F_D^*}}\sum_{O\in F_D^*}\sum_{s \in S} c^\prime_s I_s(O) = \sum_{s \in S}c^\prime_s Q_s(O),\\
& \text{where}\quad 
Q_s(D) := \frac{1}{\setsize{F_D^*}}\sum_{O\in F_D^*}I_s(O) = \Prob{O\in F_D^*}{O \supseteq s}.
\end{align*}
It suffices to show that $\deg(Q_s(D)) \le \setsize{\dom(s)}$ for every $s\in S$.



We can assume that the identity $0^n$ is in $\dom(s)$ for every partial 
function $s$ by modifying a given algorithm $A$ as follows.
At the beginning, $A$ makes the query $0^n$ 
with the initial state $\ket{0^n}\ket{(0^n)^D}\ket{0^m}$,
stores $O(0^n)$ in the answer register, 
and swaps the answer register with a part of the working register.
Afterwards, $A$ applies the original operations to 
the zero-cleared registers except for the part that stores $O(0^n)$.
Then, every $s\in S$ contains $0^n$ in its domain, and 
the modified algorithm keeps the original accepting probability and 
has the number $T+1$ of queries if the original is $T$. Therefore,
we can obtain a lower bound of $T$ from the modified algorithm.

Let 
\[
 A^i:=\set{a^{i,j}: \exists \ell\in I, s(a^{i,j}) = \sigma_\ell s(a^{i,1})}
\]
and 
\begin{align*}
    \dom(s) := \left\{
    \begin{matrix}
        a^{1,1},&...,& a^{1,v_1}\in A^1\\
        a^{2,1},&...,& a^{2,v_2}\in A^2\\
        &\vdots& \\
        a^{w,1},&...,& a^{k,v_w}\in A^w
    \end{matrix}
    \right \},
\end{align*}
where $a^{1,1}:=0^n$. 

By Lemma~\ref{lem:equiv}, 
we observe that $x^\prime=x+k_\ell\ (k_\ell\in K) \leftrightarrow O(x^\prime)=\sigma_i O(x)$
for every $i\in I$ and every $O\in F_D^*$ that hides $K$.
Since $O(a^{i,j})=\sigma_\ell O(a^{i,1}) \leftrightarrow a^{i,j}=a^{i,1}+k_\ell
\leftrightarrow a^{i,j}-a^{i,1}=0^n+k_\ell \leftrightarrow O(a^{i,j}-a^{i,1})=\sigma_\ell O(0^n)$,
$O(a^{i,j})=\sigma_\ell O(a^{i,1})$ if and only if $O(a^{i,j}-a^{i,1})=\sigma_\ell O(0^n)$ 
for every $i,j$, every $\ell\in I$ and every $O\in F_D^*$.

Then, we modify $s$ into another partial function $\tilde{s}$ by modifying $s$ as follows.
Let $s(a^{i,j})=\sigma_\ell s(a^{i,1})$ for some $\ell\in I$.
We set $\tilde{s}(a):=s(a)$ for every $a\in\dom(s)\setminus\set{a^{i,j}}$.
Since $0^n\in\dom(s)$, we can also set $\tilde{s}(a^{i,j}-a^{i,1}):=\sigma_\ell s(0^n)$.
Note that $\dom(\tilde{s})=(\dom(s)\setminus\set{a^{i,j}})\cup\set{a^{i,j}-a^{i,1}}$, and hence, $\setsize{\dom(s)}=\setsize{\dom(s^\prime)}$.
From the modification, $O$ extends $s$ if and only if $O$ extends $\tilde{s}$, and thus,
we can analyze the probability that $O$ extends $\tilde{s}$ instead of $s$.

From the above modification, we can suppose that $\dom(s) = A^1\cup A^2\cup\cdots\cup A^w$ has the following form without loss of generality.
\begin{align*}
    \dom(s) = \left\{
    \begin{matrix}
        a^{1,1},&...,& a^{1,v_1} & \in A^1,\\
        a^{2,1}& & & \in A^2,\\
        &\vdots& & &  \\
        a^{w,1}& & &  \in A^w
    \end{matrix}
    \right \},
\end{align*}
where $a^{1,1}:=0^n$.

Let $K^\prime:=\langle A^1 \rangle$ and let $D^\prime:= \setsize{K^\prime}=p^{d^\prime}$ for some $d^\prime$.
For $O\in F_D^*$ that hides $K$, let
\[
 \mathcal{E}(O) \equiv \bigg[ \bigwedge_{i=1}^{v_1} \exists \ell_i\in I: O(a^{1,i})=\sigma_{\ell_i}O(0^n) \bigg].
\]
We define 
\begin{align*}
  Q_s^R(D)  = \Prob{O\in F_D^*}{\mathcal{E}(O)},\quad
  Q_s^C(D)  = \Pr_{O\in F_D^*}\Big[O \supseteq s\, \Big|\, \mathcal{E}(O) \Big].
\end{align*}
Note that $Q_s(D) = Q_s^R(D)\cdot Q_s^C(D)$ since $\mathcal{E}(O)$ holds if $O \supseteq s$.

Since $\deg(Q_s(D))=\deg(Q_s^R(D))+\deg(Q_s^C(D))$, it suffices to estimate $\deg(Q_s^R(D))$ and $\deg(Q_s^C(D))$,
which are given in Lemmas~\ref{lem:const} and \ref{lem:injective}.

We first estimate the degree of $Q_s^R(D)$.
\begin{lemma}\label{lem:const}
We have $\deg(Q_s^R(D)) \le v_1-1$.
\end{lemma}
\begin{proof}
We observe that $\mathcal{E}(O)$ holds if and only if $K^\prime\le K$ for the subgroup $K$ hidden by $O$,
as shown below. Suppose that $\mathcal{E}(O)$ holds. From the definition of $F_D^*$,
there exists $\ell\in I$ such that $O(a^{1,i})=\sigma_{\ell_i} O(0^n)$ for every $i$ 
if and only if there exists $\ell_i\in I$ such that $a^{1,i}=k_{\ell_i}+0^n$ for every $i$.
Therefore, we obtain $a^{1,i}=k_{\ell_i}\in K$ for every $i$.

Conversely, suppose that $a^{1,i}\in K$ for every $i$ (or equivalently, $K^\prime\le K$).
From the definition of $F_D^*$, for every $x^\prime,x\in\doubleZ_p^n$, $x^\prime=x+a^{1,i}$ if and only if 
$O(x^\prime)=\sigma_{\ell_i}O(x)$ for some $\ell_i\in I$. Setting $x^\prime=a^{1,i}$ and $x=0^n$, we obtain 
$O(a^{1,i})=\sigma_{\ell_i}O(0^n)$ for every $i$.

The number of the subgroups of order $D$ that contain $K^\prime$ is equal to 
that of the subgroups of order $D/D^\prime$ of $\doubleZ_p^n/K^\prime$, 
which is isomorphic to $\doubleZ_p^{n-d^\prime}$. 
The number of subgroups is provided by the following lemma shown in \cite{KNP07}.
\begin{lemma}[Koiran et al.~\cite{KNP07}]\label{lem:group-size}
Let $n$ and $k$ be non-negative integers, and let $p$ be a prime. 
The additive group $\doubleZ_p^n$ 
has exactly $\beta_p(n,k)$ distinct subgroups of order $p^k$, where
\[
    \beta_p(n,k)=\prod_{0 \leq i < k} \frac{p^{n-i}-1}{p^{k-i}-1}.
\]
\end{lemma}

By Lemma~\ref{lem:group-size}, we have
\begin{align*}
 \Prob{O\in F_D^*}{\mathcal{E}(O)} 
  = \frac{\beta_p(n-d^\prime,d-d^\prime)}{\beta_p(n,d)}
  = \prod_{0 \leq i < d^\prime} \frac{D/p^{i}-1}{p^{n-i}-1}
\end{align*}
since $O\in F_D^*$ hides one of $\beta_p(n,d)$ subgroups 
of order $D$ uniformly at random.
Thus, its degree is at most $d^\prime\le v_1-1$.
\end{proof}

We next estimate the degree of $Q_s^C(D)$.
\begin{lemma}\label{lem:injective}
We have $\deg(Q_s^C(D)) \le w$.
\end{lemma}
\begin{proof}
Consider oracles $O\in F_D^*$ for which $\mathcal{E}(O)$ holds.
We can then define the quotient $R$ on $\doubleZ_p^n/K^\prime$ of the oracle $O\in F_D^*$.
We also define the quotient $t$ on $\doubleZ_p^n/K^\prime$ of the partial function $s$ if it exists.
Otherwise, no oracle $O$ extends such an $s$, and hence, $\deg(Q_s(D))=0$.
Thus, we suppose that $t$ exists.
Then, $s$ is extended by $O$ that hides $K$ of order $D$ if and only if 
$t$ is extended by $R$ that hides a subgroup $H:=K/K^\prime$ of order $E:=D/D^\prime=p^{d-d^\prime}$. 

More explicitly, we consider the coset decomposition of $K$ and $\doubleZ_p^n$:
$K=\cup_{i<E} \set{c_i^\prime + K^\prime}$ and 
$\doubleZ_p^n=\cup_{i<N/D}\set{c_i + K}=\cup_{i<N/D, j<E}\set{c_i+c_j^\prime + K^\prime}$,
where $c_0^\prime = 0^n$ and $N=p^n$. We then define $R(c_i+c_j^\prime):=O(c_i+c_j^\prime)$.
Also, we have $H=\set{c_j^\prime}_{j<E}$.

Let $\mathcal{E}_1(R)$ and $\mathcal{E}_2(R)$ be the events that $a^{1,1},\ldots,a^{w,1}$ are in distinct cosets of $H$ 
and that $R(a^{1,i})=t(a^{1,i})$ for every $i\in\set{1,\ldots,v_1}$, respectively.
Note that $R \supseteq t$ if and only if $\mathcal{E}_1(R)$ and $\mathcal{E}_2(R)$ hold.
Namely, $\prob{}{R\supseteq t}=\nu_t(E)\cdot\lambda_t(E)$, 
where $\nu_t(E):=\cprob{}{\mathcal{E}_2(R)}{\mathcal{E}_1(R)}$ and $\lambda_t(R):=\prob{}{\mathcal{E}_1(E)}$.

We can complete the proof of Lemma~\ref{lem:injective} 
by Lemmas~\ref{lem:nu} and \ref{lem:lambda}.
The proofs of the two lemmas are very similar to the corresponding ones given in \cite{KNP07}.

\begin{lemma}\label{lem:nu}
We have $\deg(\nu_t(E)) = 0$.
\end{lemma}
\begin{proof}
Since $\mathcal{E}(R)$ holds, $a^{1,1},\ldots, a^{w,1}$ are in distinct coset of $H$.
Also, we have $R(c_i)=O(c_i)$ from the definition of $R$.
Let $B:=\set{a^{1,1},\ldots,a^{w,1}}\cap\set{c_i}_{i<N/D}$
and $\bar{B}:=\set{a^{1,1},\ldots,a^{w,1}}\setminus B$. Let $w^\prime:=\setsize{B}$.

From the definition of $\Pi_K$, once $O_0(c_0),\ldots,O_0(c_{(N/D)-1})$ are fixed,
all the output values of $O$ are uniquely determined, and so is $R$. Therefore, 
the number of possible $R$ is that of possible values assigned to $R_0(c_0),\ldots,R_0(c_{(N/D)-1})$,
which equals $p^n(p^n-1)\cdots(p^n-(p^{d-d^\prime}-1))$.

Since $R(a)$ is fixed for every $a\in B$ if $R\supseteq t$,
the number of possible $R$ that extends $t$ is 
that of possible values assigned to $\set{R(a)}_{a\in \bar{B}}$, which equals 
$(p^n-w^\prime)(p^n-w^\prime-1)\ldots(p^n-(p^{d-d^\prime}-1))$.
Therefore, we have $\nu_t(E)=\{p^n(p^n-1)\cdots(p^n-(w^\prime-1))\}^{-1}$.
The degree of $\nu_t$ is $0$.
\end{proof}

\begin{lemma}\label{lem:lambda}
We have $\deg(\lambda_t(E))\le w$.
\end{lemma}
\begin{proof}
From the definition of $\lambda_t(E)$, 
the probability that $H$ contains at least one $a_i-a_j \;(i \neq j)$ is $1-\lambda_t(E)$. 
Therefore, from the inclusion-exclusion principle, we have
\begin{align*}
1-\lambda_t(E) & = \Prob{}{\exists i\ne \exists j, a^{i,1}-a^{j,1} \in H}\\
 & = \sum_{i\neq j}\Prob{}{a^{i,1}-a^{j,1} \in H} - \cdots + \Prob{}{\forall i \neq \forall j, \, a^{i,1} - a^{j,1} \in H}.
\end{align*}

Note that $a^{i,1}-a^{j,1} \in H$ if and only if there exists $\ell\in J$ 
such that $O(a^{i,1})=\sigma_\ell O(a^{j,1})$, where $J \subseteq I$ and $k_\ell\in H$ for every $\ell\in J$.
From the same argument used above 
each term in the above sum is a polynomial in $E$ of degree at most $e^\prime$, 
where the order of the subgroup generated by $\set{a^{i,1}-a^{j,1}}_{i\neq j}$ is $p^{e^\prime}\le p^w$. Therefore, the degree of $\lambda_t$ is at most $w$.
\end{proof}

From Lemmas~\ref{lem:nu} and \ref{lem:lambda}, we obtain $\deg(Q_s^C(D))\le w$ (Lemma~\ref{lem:injective}). 
\end{proof}

We obtain $\deg(Q_s(D))\le v_1-1+w \le \max_{s\in S}\setsize{\dom(s)}$ (Lemma~\ref{lem:domain-size})
from Lemmas~\ref{lem:const} and \ref{lem:injective}.
\end{proofof} 

Therefore, We can prove the main theorem (Theorem~\ref{thm:main}) from Lemma~\ref{lem:domain-size}.
\end{proofof} 

\section{Concluding Remarks}\label{sec:concl}
The oracle distribution (that is uniform over $F_D^*$) used 
for the quantum query lower bounds is artificially biased 
because of the condition ``$O\in \Pi_K$'' in the definition of $F_D^*$.
This condition is crucial in the proof of Lemma~\ref{lem:nu}.
It is natural to use the uniform distribution over $F_D$ to 
prove the average-case lower bounds, but the polynomial method fails
because $\nu_t$ could contain some term exponential 
in $E=D/D^\prime$ when $O\in F_D$. 
Hence, we need new proof techniques for quantum query lower bounds 
in the natural average case.

The obvious open problem is to prove the quantum security of 
classically secure variants of the EM cipher such as 
Iterated EM cipher \cite{CS14} and SoEM \cite{CLM19}, 
but there seem to be no approaches to them so far.
The algebraic characterization of the oracle used in this paper 
could help to establish security proofs for quantum adversaries.
\section*{Acknowledgments}
This work was supported by JSPS Grant-in-Aid for Scientific Research (A) Nos.~21H04879, 23H00468, (C) No.~21K11887, JSPS Grant-in-Aid for Challenging Research (Pioneering) No.~23K17455, and MEXT Quantum Leap Flagship Program (MEXT Q-LEAP) Grant Number JPMXS0120319794. 

\bibliographystyle{plain}
\bibliography{ourbib}


\end{document}